\documentclass[11pt,reqno]{amsart}

\calclayout
\let \oldsection \section
\renewcommand{\section}{\vspace{8pt}\oldsection}

\usepackage[T1]{fontenc}
\usepackage[latin9]{inputenc}
\usepackage{graphicx}
\usepackage{amsmath,amsthm,amssymb,amsfonts}
\usepackage{xcolor}
\usepackage{doi}
\hypersetup{citecolor=purple,linkcolor=blue,urlcolor=teal}

\newcommand{\ket}[1]{{|{#1}\rangle}}
\newcommand{\bra}[1]{{\langle{#1}|}}

\newcommand{\C}{\mathbb{C}}
\newcommand{\cA}{\mathcal{A}}

\newcommand{\cD}{\mathcal{D}}

\newcommand{\cH}{\mathcal{H}}
\newcommand{\cK}{\mathcal{K}}

\newcommand{\cR}{\mathcal{R}}
\newcommand{\cX}{\mathcal{X}}

\newcommand{\zo}[1]{\{0,1\}^{#1}}
\DeclareMathOperator*{\Exp}{\mathbb{E}} 
\DeclareMathOperator{\Span}{\operatorname {span}}
\DeclareMathOperator{\poly}{\operatorname {poly}}
\DeclareMathOperator{\negl}{\operatorname {negl}}
\DeclareMathOperator{\PRP}{\operatorname {PRP}}

\newtheorem{theorem}{Theorem}[section]
\newtheorem{corollary}[theorem]{Corollary}

\newtheorem{lemma}[theorem]{Lemma}
\newtheorem{fact}[theorem]{Fact}
\newtheorem{definition}[theorem]{Definition}
\newtheorem{prop}[theorem]{Proposition}
\numberwithin{equation}{section}

\newcommand{\trace}{\mathrm{Tr}}

\title{Pseudorandom and Pseudoentangled States from Subset States}

\makeatletter 
\newcommand\thankssymb[1]{\textsuperscript{\@fnsymbol{#1}}}
\makeatother

\author{Fernando Granha Jeronimo\thankssymb{1}}
\thanks{\thankssymb{1} {\tt IAS \& Simons Institute}. {\tt granha@ias.edu}. Supported as a Google Research Fellow.} 
\author{Nir Magrafta\thankssymb{2}}
\thanks{\thankssymb{2} {\tt Weizmann Institute of Science}. {\tt nir.magrafta@weizmann.ac.il}. Supported by the Israel Science Foundation (Grant No.\ 3426/21), and by the European Union Horizon 2020 Research and Innovation Program via ERC Project REACT (Grant 756482)}
\author{Pei Wu\thankssymb{3}}
\thanks{\thankssymb{3} {\tt Weizmann Institute of Science}. {\tt pei.wu@weizmann.ac.il}.  Supported by ERC Consolidator Grant VerNisQDevS (101086733).}

\begin{document}

\maketitle

\begin{abstract}
Pseudorandom states (PRS) are an important primitive in quantum cryptography. In this paper, we show that \emph{subset states} can be used to construct PRSs. A subset state with respect to $S$, a subset of the computational basis, is 
\[
    \frac{1}{\sqrt{|S|}}\sum_{i\in S} |i\rangle.
\] 
As a technical centerpiece, we show that for any fixed subset size $|S|=s$ such that $s = 2^n/\omega(\poly(n))$ and $s=\omega(\poly(n))$, where $n$ is the number of qubits, a random subset state is information-theoretically indistinguishable from a Haar random state even provided with polynomially many copies.
This range of parameter is tight. Our work resolves a conjecture by Ji, Liu and Song~\cite{ji2018pseudorandom}. 
Since subset states of small size have small entanglement  across all cuts, this construction also illustrates a pseudoentanglement phenomenon.
\end{abstract}

\section{Introduction}

Pseudorandom quantum states (PRS) are a keyed family of quantum states that can be efficiently generated and are computationally indistinguishable from Haar random states, even when provided with polynomially many copies.
PRSs have a wide range of applications including but not limited to statistically binding quantum bit commitments \cite{morimae2022quantum} and private quantum coins \cite{ji2018pseudorandom}. Notably, for certain applications like private quantum coins, PRSs represent the weakest primitive known to imply them.
Moreover, PRSs imply other quantum pseudorandom objects, such as one-way state generators (OWSGs) \cite{ji2018pseudorandom,morimae2022one} and EFI pairs (\textbf{e}fficiently samplable, statistically \textbf{f}ar but computationally \textbf{i}ndistinguishable pairs of quantum states) \cite{brakerski2022computational,morimae2022quantum}.
Although all the existing PRS constructions rely on quantum-secure pseudorandom functions (PRFs) or pseudorandom permutations (PRPs),  PRSs may be weaker than PRFs~\cite{kretschmer2023quantum}.

Since the initial proposal of pseudorandom states~\cite{ji2018pseudorandom}, various constructions have been investigated. Ji et al. pioneered by constructing PRS with complex phases over the uniform superposition~\cite{ji2018pseudorandom}. Subsequently, it is proved that using only $\pm 1$ phases is sufficient for PRS construction~\cite{brakerski2019pseudo, ananth23binary}. Varying magnitudes as well as phases led to scalable PRS~\cite{brakerski2020scalable}, PRS with provable high and low entanglement entropy \cite{aaronson2024quantum} and PRS with proofs of destruction~\cite{behera2023pseudorandomness}. 

Randomizing the phase was essential in the security proofs for all of these constructions.
It is then natural to ask if it is possible to construct PRS without varying the phases. Indeed, \cite{ji2018pseudorandom} have conjectured that random \emph{subset states} can be used to construct PRS.

A closely related notion to the PRSs is that of \emph{pseudoentangled} states studied recently by~\cite{aaronson2024quantum}. Here we call a PRS $h(n)$-pseudoentangled if for any state $\ket\phi$ from the PRS, $\ket\phi$ in addition satisfies that its entanglement entropy across all cut is $O(h(n)).$\footnote{In \cite{aaronson2024quantum}, another notion was considered. Roughly speaking, a pseudoentangled state ensemble consists of two efficiently constructible and computationally indistinguishable ensembles of states which display a gap in their entanglement entropy across all cuts.} Note that for a Haar random state, the entanglement entropy is near maximal across all cuts.
It's observed in~\cite{aaronson2024quantum} that a subset state with respect to set $S$ has entanglement entropy at most $O(\log |S|)$ across any cut for some function $h(n):\mathbb{N}\to\mathbb{N}$, since the Schmidt rank of a subset state is at most $|S|$ across any cut. Therefore the subset states with small set size are good candidates for pseudoentangled states, which they left as an open problem.

\subsection*{Our results} In this paper, we settle the conjecture and prove that random subset states and Haar random states are information-theoretically indistinguishable, even with polynomially many copies:
\begin{theorem}\label{thm:subset-state-PRS}
    Let $\cH=\C^d$ be a Hilbert space of dimension $d \in \mathbb{N}$, $\mu$ be the Haar measure on $\mathcal{H}$, and $S\subseteq [d]$ of size $s$. Then for any $k\in \mathbb{N}$,
    \begin{align}
        \left\| \int{\psi^{\otimes k} d\mu(\psi)} - \mathop{\mathbb{E}}_{S\subseteq[d], |S|=s} \phi_S ^{\otimes k} \right\|_1 \le O\left(\frac{k^2}{d} + \frac{k}{\sqrt{s}} + \frac{s k}{d}\right),
    \end{align}
where 
\[
    \phi_S =  \left(\frac{1}{\sqrt{s}}\sum_{i\in S}\ket i\right)\left(\frac{1}{\sqrt{s}}\sum_{i\in S} \bra i \right).
\]
\end{theorem}

Consider a  $n$-qubit system, represented by a $2^n$ dimensional Hilbert space.
For any function $t(n)=\omega(\poly(n))$ and $t(n) \le s \le 2^n/t(n)$ and $k=\poly(n)$, the distance between $k$ copies of a Haar random state and $k$ copies of a random subset state of size $s$ is negligible as per the above theorem. This range for the subset's size is tight, as otherwise  efficient distinguishers exist between copies of a Haar random state and a random subset state.

An immediate corollary of the above theorem is the following:
\begin{corollary}[Pseudorandom States]\label{cor:sub-state-are-prs-informal} 
    Let $\{\PRP_k\}_{k\in \cK}$ be a quantum-secure family of pseudorandom permutations. Then the family of states 
    $$\left\{\frac{1} {\sqrt{s}}\sum_{x\in {[s]}} \ket{\PRP_k(x)} \right\}_{k\in \cK} $$ is a PRS on $n$ qubits for $t(n) \le s \le 2^n/t(n)$ and any $t(n)=\omega(\poly(n))$.
\end{corollary}

The state $\sum_{x\in {[s]}}  \ket{\PRP_k(x)} / \sqrt{s}$  can be efficiently prepared: First, compute the superposition $\sum_{x\in {[s]}} \ket{x} / \sqrt{s}$; then compute $\sum_{x\in {[s]}} \ket{x, \PRP_k(x)} / \sqrt{s}$ with the classical circuit of $\PRP_k$; finally uncompute the first register by running the circuit  $\PRP^{-1}_k$.
The pseudorandomness property of the above PRS construction follows  from
the definition of quantum-secure pseudorandom permutations (see in Section~\ref{sec:prs}) and Theorem~\ref{thm:subset-state-PRS}.

Since a subset state with respect to a small subset has small entanglement entropy across all cut, we also have
\begin{corollary}[Pseudoentangled States]
    Let $\PRP_k:[2^n]\to[2^n]$ be a quantum-secure family of pseudorandom permutations. For any $h(n)= \omega(\log n)$ and $h(n)=n-\omega(\log n)$, we have the following $h(n)$-pseudoentangled state from subset state of size $s=2^{h(n)}$, $$\left\{\frac{1}{\sqrt{s}}\sum_{x\in {[s]}} \ket{\PRP_k(x)}\right\}_{k\in\cK}.$$
\end{corollary}
For a PRS, it is easy to see that if for some cut the entanglement entropy of a state $\ket\phi$ is $O(\log n)$, then $\ket\phi$ can be distinguished from Haar random states with polynomially many copies using swap test. 

\subsection*{Contribution} This is the first proved PRS construction that requires only a single phase, showing that varying phases is unnecessary for state pseudorandomness. This is especially interesting when contrasting with the recent result by \cite{haug2023pseudorandom} that shows pseudorandom unitaries require imaginary parts. 
Our main insight lies in making the connection to the study of graph spectra. In our analysis of the trace distance in Theorem~\ref{thm:subset-state-PRS}, we identify  the matrix as a weighted sum of adjacency matrices corresponding to generalized Johnson graphs, of which the spectra are well understood. The proof then is straightforward, relying on in addition only the small collision probability of sampling polynomially many elements from a superpolynomial-size set.

\subsection*{Concurrent work}
Concurrently and independently, Giurgica-Tiron and Bouland proved a similar result \cite{bouland23subset}. Both works build upon related results about Johnson graphs and the Johnson scheme, and \cite{bouland23subset} mention the connection in their discussion. While we find our proof to be simpler to understand, their bound is slightly stronger: $O\left(sk/d + k^2/s\right)$. Additionally, it was brought to our attention that Fermi Ma also proved a similar result concurrently and independently. 

\section{Preliminaries}
\subsection{Notations}

Given any matrix $M\in \C^{n\times n}$ denote by $\|M\|_1$ the trace norm, which is the sum of the singular values of $M$.
We write $x\lesssim y$ to denote that there is a small constant $c\ge 1$, such that $x\le c y$. 
For any set $S$, let
$A(S,k):=\{(i_{1},i_{2},\ldots,i_{k})\in S{}^{k}:i_{j}\not=i_{j'}\text{ for }j\not=j'\}.$ We also adopt the notation $S_n$ for the symmetric group.

Let $n^{\underline{k}}:=n(n-1)\cdots(n-k+1).$ 
A simple calculation reveals that for $k=o(d)$,
\begin{equation*}
    \frac{d^{\underline k}}{ (d-k+1)^{\underline k}} -1 \le  \left(\frac{d+k-1}{d-k+1}\right)^k -1
    =
    O\left(\frac{k^2}{d}\right).
\end{equation*}
We will use this bound without referring to this calculations. Given any pure quantum state $\rho$ over systems $A,B$. The entanglement entropy over the cut $A:B$ is the von Neumann entropy of the reduced density matrix of system $A$ (or $B$), i.e., $-\trace (\rho_A \log \rho_A)$.

\subsection{PRPs and PRSs}\label{sec:prs}
Pseudorandom permutations (PRPs) are important constructions in cryptography. They are families of efficiently computable permutations that looks like truly random permutations
to polynomial-time (quantum) machines. 

\begin{definition}[Quantum-Secure Pseudorandom Permutation] Let $\cK$ be a key space, and $\cX$ the domain and range, all implicitly depending on the security parameter $\lambda$. A keyed family of permutations $\{\PRP_k\}_{k\in \cK}$ is a quantum-secure pseudorandom permutation if for any polynomial-time quantum oracle algorithm $\cA$, $\PRP_k$ with a random $k \leftarrow \cK$ is indistinguishable from a truly random permutation $P$ in the sense that
$$
    \left|
    \Pr_k [\cA^{\PRP_k, \PRP_k^{-1}} (1^\lambda)=1] - \Pr_P [\cA^{P,P^{-1}}(1^\lambda)=1]
    \right| = \negl(\lambda).
$$
In addition, $\PRP_k$ is polynomial-time computable on a classical
computer.
\end{definition}

\begin{fact} [\cite{zhandry2012construct, zhandry2016note}]
    Quantum-secure PRPs exist if quantum-secure one-way functions exist.
\end{fact}

The notion of pseudorandom quantum states is introduced in \cite{ji2018pseudorandom}. We restate it here.

\begin{definition}[Pseudorandom Quantum States]
    Let $\lambda$ be the security parameter. Let $\cH$ be a Hilbert space and $\cK$ the key space, both parameterized by $\lambda$. A keyed family of quantum states $\{\ket{\phi_k}\in \cH\}_{k\in \cK}$ is pseudorandom if the following two conditions hold:
    \begin{enumerate}
        \item (Efficient generation). There is a polynomial-time quantum algorithm $G$ that generates state $\ket{\phi_k}$ on input $k$. That is, for all $k \in \cK$, $G(k)=\ket{\phi_k}$.
        \item (Pseudorandomness). Any polynomially many copies of $\ket{\phi_k}$ with the same random $k \in \cK$ are computationally indistinguishable from the same number of copies of a Haar random state. More precisely, for any  polynomial-time quantum algorithm $\cA$ and any $m\in \poly(\lambda)$,
        $$
        \left|
        \Pr_{k\leftarrow \cK} [\cA^{}(\ket{\phi_k}^{\otimes m})=1] -
        \Pr_{\ket{\psi}\leftarrow \mu} [\cA^{}(\ket{\psi}^{\otimes m})=1]
        \right|=\negl(\lambda)
        $$
        where $\mu$ is the Haar measure on $\cH$.
    \end{enumerate}
\end{definition}

\section{A Random Subset State is Indistinguishable from A Haar Random State}
In this section, we prove Theorem~\ref{thm:subset-state-PRS}. The proof involves three simple propositions.
First, let's look at the Haar random state. 
A well-known fact by representation theory gives an explicit formula
for the mixture of Haar random states, $\Psi = \int \psi^{\otimes k} d \mu $, where $\mu$ is the Haar measure. For a detailed proof, see for example \cite{harrow2013church}.
\begin{fact}
\begin{align*}
\int \psi^{\otimes k} d \mu & =\binom{d+k-1}{k}^{-1}\cdot\frac{1}{k!}\sum_{\pi\in S_{k}}\sum_{\vec{i}\in[d]^{k}}|\vec{i}\rangle\langle\pi(\vec{i})|.
\end{align*}
\end{fact}
Instead of working with $\Psi$ directly, we look at the operator $\tilde \Psi = \Pi \Psi \Pi$, where $\Pi$ is the projection onto the subspace of
\[ \Span\{\ket{\vec{i}} : \vec i \in A([d],k)\} \subseteq \cH^{\otimes k}.\]
Immediately,
\begin{align}
\tilde{\Psi}=\binom{d+k-1}{k}^{-1}\cdot\frac{1}{k!}\sum_{\pi\in S_{k}}\sum_{\vec{i}\in A([d],k)}|\vec{i}\rangle\langle\pi(\vec{i})|.
\end{align}
As long as $k$ is small, we expect that $\Psi \approx \tilde \Psi$. This is simple and known. For completeness, we present a proof.

\begin{prop}\label{prop:haar-pi-haar}
$\|\Psi -\tilde\Psi\|_{1}=O(k^{2}/d).$
\end{prop}

\begin{proof}
Consider the following decomposition of $\Psi:=\tilde\Psi +\cR$. Note that 
\[
\cR = (I-\Pi) \Psi (I-\Pi).
\]
It's clear that $\tilde{\Psi}$ and $\cR$ are both positive semi-definite, and
$\tilde{\Psi}\cR= 0.$ Therefore, the nonzero eigenspaces of $\Psi$ correspond to those of $\tilde \Psi$ and $\cR$, respectively. Consequently,
\begin{align*}
\left\Vert \Psi-\tilde{\Psi}\right\Vert _{1} & =1-\|\tilde{\Psi}\|_{1}=1-\frac{d^{\underline{k}}}{(d+k-1)^{\underline{k}}}\\
 & \qquad=1-\frac{d}{d+k-1}\cdot\frac{d-1}{d+k-2}\cdots\frac{d-k+1}{d}\\
 & \qquad\le O\left(\frac{k^{2}}{d}\right).\qedhere
\end{align*}
\end{proof}

Next, we turn to random subset states. Let $\Phi = \Exp_{|S|=s} \phi_S ^{\otimes k}$, and consider $\Pi \Phi \Pi$, but normalized.\footnote{Although in the case of Haar random state we didn't normalize, this doesn't really matter. Our choice is for simplicity of proof.} In particular,
\[
\tilde{\Phi}=\Exp_{S:|S|=s}\left[\frac{1}{s^{\underline{k}}}\sum_{\vec{i},\vec{j}\in A(S,k),}|\vec{i}\rangle\langle\vec{j}|\right].
\]
Analogous to Proposition~\ref{prop:haar-pi-haar}, we have
\begin{prop}\label{prop:subset-pi-subset}
$\|\Phi-\tilde{\Phi}\|_{1}=O(k/\sqrt{s}).$
\end{prop}

\begin{proof} Let $\gamma$ be the uniform distribution over subset $S\subseteq [d]$  of size $s$,
\begin{align*} 
 & \left\Vert \int_S \left(\frac{1}{s^{k}}\sum_{\vec{i},\vec{j}\in S^{k}}|\vec{i}\rangle\langle\vec{j}|-\frac{1}{s^{\underline{k}}}\sum_{\vec{i},\vec{j}\in A(S,k)}|\vec{i}\rangle\langle\vec{j}|\right) d\gamma \right\Vert _{1}\\
 & \qquad\le\int_{S}\left\Vert \frac{1}{s^{k}}\sum_{\vec{i},\vec{j}\in S^{k}}|\vec{i}\rangle\langle\vec{j}|-\frac{1}{s^{\underline{k}}}\sum_{\vec{i},\vec{j}\in A(S,k)}|\vec{i}\rangle\langle\vec{j}|\right\Vert _{1} d\gamma\le O\left(\frac{k}{\sqrt{s}}\right).\qedhere
\end{align*}
\end{proof}
%
All that is left to do is to show that $\|\tilde{\Phi}-\tilde{\Psi}\|_{1}$ is small.
Fix any $\vec{i},\vec{j}\in A([d],k)$, and let $\ell=\ell(\vec{i},\vec{j})$
be the total number of distinct elements in the union of the elements of the vectors  $\vec{i},\vec{j}$. Then the $(\vec{i},\vec{j})$'th entry of $\tilde{\Phi}$ is
\begin{align}
\tilde{\Phi}(\vec{i},\vec{j}) & =\frac{1}{s^{\underline{k}}}\Pr_{|S|=s}[\vec{i},\vec{j}\in A(S,k)]=\frac{1}{s^{\underline{k}}}\frac{\binom{d-\ell}{s-\ell}}{\binom{d}{s}}=\frac{s^{\underline{\ell}}}{s^{\underline{k}}\cdot d^{\underline{\ell}}}.\label{eq:entry-value}
\end{align}

\begin{prop}\label{prop:haar-vs-subset-appr}
For any $k\ll s\le d,$ it holds that \[\|\tilde{\Phi}-\tilde{\Psi}\|_{1}=O\left(\frac{sk}{d}\right).\]
\end{prop}

\begin{proof}
Let 
\[
\cD=\tilde{\Phi}-\frac{(d+k-1)^{\underline{k}}}{d^{\underline{k}}}\tilde{\Psi}.
\]
This factor is chosen so that $\cD(\vec{i},\vec{j})=0$ for any $\vec{i}$
and $\vec{j}$ such that $\vec{j}=\pi(\vec{i})$ for some permutation
$\pi.$ By triangle inequality,
\[
\|\tilde{\Phi}-\tilde{\Psi}\|_{1} \le 
    \left\|
    \cD
    \right\|_1
    + \left\|
            \tilde{\Psi}-\frac{(d+k-1)^{\underline{k}}}{d^{\underline{k}}}\tilde{\Psi}
    \right\|_1,
\]
where the second term is bounded by $O(k^2/d)$.

We turn to $\cD$. Let $\vec{j}\sim\vec{i}$ to denote that $\vec{j}$ is a permutation
of $\vec{i}$. Note that for any $\vec{i}\sim\vec{j}$, $\cD(\vec{i},\cdot)=\cD(\vec{j},\cdot)$,
and similarly $\cD(\cdot,\vec{i})=\cD(\cdot,\vec{j}).$ Therefore
$\cD=\tilde{\cD}\otimes J$
where $J\in\C^{k!\times k!}$ is the all 1 matrix and $\tilde{\cD}\in\C^{\binom{[d]}{k}\times\binom{[d]}{k}},$
s.t. for any $A,B\in\binom{[d]}{k},$ 
\begin{align*}
\tilde{\cD}(A,B)=\cD(\vec{i},\vec{j}), &  & \vec{i},\vec{j}\text{ contain }A\text{ and }B,\text{ respectively.}
\end{align*}
Next, decompose $\tilde{\cD}:=\sum_{t=0}^{k-1}\alpha_{t}\cD_{t},$ where in view of (\ref{eq:entry-value}),
\begin{align*}
 & \alpha_{t}=\frac{(s-k) \cdots (s-2k+t+1)}{d \cdots (d-2k+t+1)},\\
 & \cD_{t}(A,B)=\begin{cases}
1 & |A\cap B|=t,\\
0 & \text{otherwise.}
\end{cases}
\end{align*}
$\cD_{t}$ is the adjacency matrix for the well-studied generalized
Johnson graphs~\cite{delsarte1973algebraic}. In particular,
we will need the following fact (explained in Appendix).
\begin{fact}\label{fact:johnson-trace}
For any $0\le t\le k-1$, and for $k=O(\sqrt{d})$
$$
\|\cD_{t}\|_{1}\lesssim \binom{d-k}{k-t}\binom{d}{t}2^{k-t}.
$$
\end{fact}
Assisted by the above fact, we can bound $\|\cD\|_{1}$ for $sk=O(d)$ as below, 
\begin{align*}
    \|\cD\|_{1}=k! & \|\tilde{\cD}\|_{1}\le k!\sum_{t=0}^{k-1}\alpha_{t}\|\cD_{t}\|_{1}\lesssim k!\sum_{t=0}^{k-1}\frac{s^{k-t}}{d^{2k-t}}\cdot\frac{d^{k}}{t!(k-t)!}\cdot2^{k-t}\\
     & =\sum_{t=0}^{k-1}\left(\frac{2s}{d}\right)^{k-t}\binom{k}{t}= \left(1+\frac{2s}{d}\right)^{k}-1\lesssim O\left(\frac{2sk}{d}\right). \qedhere
\end{align*}
\end{proof}

Theorem~\ref{thm:subset-state-PRS} follows by triangle inequality on Propositions ~\ref{prop:haar-pi-haar}-\ref{prop:haar-vs-subset-appr}.

\section{Tightness of the set size}
In this section, we discuss the optimality of Theorem~\ref{thm:subset-state-PRS} in terms of subset size.
\begin{lemma}
    For any polynomial $p$, there are polynomial-time quantum algorithms that distinguish a Haar random state with any subset state of size $s$ provided polynomially many copies for either of the following two cases
    \begin{enumerate}
        \item\label{enu:small-set} $s = p(n)$,
        \item\label{enu:large-set} $s = 2^n / p(n).$
    \end{enumerate}
\end{lemma}
\begin{proof}
    (\ref{enu:small-set}) For a  subset of small size $p(n)$, measure $p(n)+1$ copies in the standard basis. Return ``random subset state'' if two outcomes are equal, and ``Haar random'' otherwise.
    By the pigeonhole principle we get some outcome twice in case we have copies of a random subset state, answering correctly with probability $1$.
    On the other hand, Haar random states are exponentially close to random binary phase states in trace distance \cite{ji2018pseudorandom,brakerski2019pseudo}. These have equal inverse exponential magnitude for each computational basis element. Thus the probability of getting the same outcome twice is exponentially small. 

    (\ref{enu:large-set}) For large set size of size $2^n/p(n)$, perform the swap test on the input state and the uniform superposition of all computational basis. The fidelity between the uniform superposition and the random subset state is $1/\sqrt{p(n)}$,
    and therefore the swap test outputs ``equal'' with probability $\frac{1}{2}\left(1+\frac{1}{p(n)}\right)$. 
    On the other hand, the fidelity between a random binary phase state and uniform superposition is exponentially small, therefore Haar random state and uniform superposition as well. In conclusion, the  test accepts subset state with an advantage at least $1/3p(n)$ than that of Haar random state. Repeating with more copies, the advantage can be amplified to $\Omega(1)$.
\end{proof}

\section*{Acknowledgement} We would like to express our gratitude to Zvika Brakerski and Omri Shmueli for their valuable feedback on our paper.

\bibliographystyle{plain}
\bibliography{Bibliography}

\appendix

\section*{Appendix: Spectra of $\cD_{t}$}

The spectra of $\cD_{t}$ is known \cite{delsarte1973algebraic}. In particular,
fix any $0\le t\le k-1,$ there are $k+1$ distinct eigenvalues $\lambda_{0},\lambda_{1},\ldots,\lambda_{k},$
such that
\begin{align*}
 & \lambda_{0}=\binom{k}{t}\binom{d-k}{k-t},\\
 & \lambda_j = \sum_{\ell=\max\{0,j-t\}}^{\min\{j, k-t\}}(-1)^\ell\binom{j}{\ell}\binom{k-j}{k-t-\ell}\binom{d-k-j}{k-t-\ell},& j = 1,2,\ldots, k,
\end{align*}
with multiplicity
\begin{align*}
& m_{0}=1,\\
 & m_{j}=\binom{d}{j}-\binom{d}{j-1}, & j=1,2,\ldots,k.
\end{align*}
For us, we simplify $\lambda_j$ for $k=O(\sqrt d)$,
\begin{align*}
 & |\lambda_{j}|\lesssim\frac{\binom{k-j}{t-j}}{\binom{k}{t}}\lambda_{0}, & j=1,2,\ldots,t,\\
 & |\lambda_{j}|\lesssim\frac{\binom{j}{t}(k-t)!}{\binom{k}{t}(k-j)!}\cdot\frac{1}{d^{j-t}}\lambda_{0}, & j=t+1,\ldots,k.
\end{align*}
We bound $\|\cD_{t}\|_{1}$ as follows 
\begin{align*}
\|\cD_{t}\|_{1} & =\lambda_{0}+\sum_{j=1}^{k}m_{j}|\lambda_{j}|\\
 & \lesssim\lambda_{0}\left(1+\sum_{j=1}^{t}\frac{d^{\underline{j}}}{j!}\frac{\binom{k-j}{t-j}}{\binom{k}{t}}+\sum_{j=t+1}^{k}\frac{d^{\underline{t}}}{j!}\frac{\binom{j}{t}(k-t)^{\underline{j-t}}}{\binom{k}{t}}\right)\\
 & \lesssim\lambda_{0}\left(1+\frac{d^{\text{\ensuremath{\underline{t}}}}}{k^{\underline{t}}}+\sum_{j=t+1}^{k}\frac{d^{\underline{t}}}{k^{\underline{t}}}\binom{k-t}{j-t}\right)\\
 & \le\lambda_{0}\left(1+\frac{d^{\underline{t}}}{k^{\underline{t}}}2^{k-t}\right)\\
 & =\binom{k}{t}\binom{d-k}{k-t}
    \left(1+ \frac{d^{\underline{t}}}{k^{\underline{t}}}2^{k-t} \right)
    \\
 &\lesssim \binom{k}{t}\binom{d-k}{k-t} \frac{d^{\underline{t}}}{k^{\underline{t}}}2^{k-t}
 =\binom{d}{t}\binom{d-k}{k-t}2^{k-t}.
\end{align*}
This proves Fact~\ref{fact:johnson-trace}.

\end{document}